\documentclass[letterpaper, 10 pt, conference]{ieeeconf}  

\IEEEoverridecommandlockouts                              
\usepackage{amsmath} 
\usepackage{amssymb}  

\usepackage{graphicx}

\newcommand{\pa}{\mathrm{pa}}
\newcommand{\anc}{\mathrm{anc}}

\newcommand{\impute}{\mathrm{do}}

\newcommand{\E}{\mathbb{E}}
\newcommand{\R}{\mathbb{R}}

\newcommand{\Xc}{\mathcal{X}}

\newcommand{\trace}{\mathrm{trace}}
\newcommand{\diag}{\mathrm{diag}}

\newcommand{\T}{\top}

\usepackage{amsthm}

\usepackage{algorithm}
\usepackage{algorithmic}

\newtheorem{assumption}{Assumption}
\newtheorem{definition}{Definition}
\newtheorem{proposition}{Proposition}

\newtheorem{theorem}{Theorem}
\newtheorem{corollary}{Corollary}

\title{\LARGE \bf
Optimal Causal Imputation for Control
}

\author{Roy Dong, Eric Mazumdar, and S. Shankar Sastry
\thanks{R. Dong, E. Mazumdar, and S. S. Sastry are with the Department of Electrical Engineering and Computer Sciences, 
        University of California, Berkeley, Berkeley, CA, 94707, USA
        {\tt\footnotesize $\{$roydong,emazumdar,sastry$\}$@eecs.berkeley.edu}}%
}

\begin{document}

\maketitle
\thispagestyle{empty}
\pagestyle{empty}


\begin{abstract}

The widespread applicability of analytics in cyber-physical systems has motivated research into causal inference methods. Predictive estimators are not sufficient when analytics are used for decision making; rather, the flow of causal effects must be determined. Generally speaking, these methods focus on estimation of a causal structure from experimental data. In this paper, we consider the dual problem: we fix the causal structure and optimize over causal imputations to achieve desirable system behaviors for a minimal imputation cost. First, we present the optimal causal imputation problem, and then we analyze the problem in two special cases: 1) when the causal imputations can only impute to a fixed value, 2) when the causal structure has linear dynamics with additive Gaussian noise. This optimal causal imputation framework serves to bridge the gap between causal structures and control.


\end{abstract}

\section{Introduction}
\label{sec:intro}


Recently, data analytics have achieved amazing levels of success. 
As analytics penetrate more and more industrial applications, they are increasingly used for decision-making and planning. In these applications, it is important to use estimators that are not only predictive, but estimate the causal structure of the underlying processes.

Correlation is not the same as causation. 
However, in practice, it is not always easy to apply this principle. 
In many real-life applications, machine learning is used to determine the relationship between two variables. This analysis is often used as the basis for determining which actions to take. However, an algorithm with low test error does not necessarily mean that the causal effect has been estimated. 

For example, one may train a classifier to estimate the energy consumption of a household given the presence and absence of eco-friendly devices, and this may provide guidelines for which devices should be discounted through rebate programs. Unless the causal structures are explicitly accounted for, there could easily be confounding variables or incorrect causal relationships that change the behavior of the system under consideration.

This has motivated new interest in causal inference techniques. Generally speaking, these techniques take experimental data and attempt to uncover the causal structure. (We defer a literature review of these methods to Section~\ref{sec:model}, when a more formal model of causality has been developed.) In this paper, we consider the dual problem: we fix the causal structure and attempt to determine what causal actions will lead to system behaviors we desire at a minimal cost.

\subsection{Outline}

The rest of the paper is organized as follows. We discuss the main paradigms for modeling causality in Section~\ref{sec:back}. In Section~\ref{sec:model}, we outline the mathematical formulation of a causal structure, discuss relevant literature in causal estimation, and define the problem of optimal causal imputation. In Section~\ref{sec:app}, we provide theoretical analysis of two special cases of the optimal causal imputation problem: the case where imputation can only be done to a single value, and the case where the dynamics are linear and the noise is Gaussian. 
Finally, we present closing remarks in Section~\ref{sec:conclusions}.

\section{Background}
\label{sec:back}


There are three main paradigms for the mathematical modeling of causality:
\begin{enumerate}
\item Rubin causality
\item Granger causality
\item Pearl's structural equation modeling (SEM)
\end{enumerate}
Each of these paradigms has a vast literature in its own right; we will try to present a few representative samples from each field here. Note that each paradigm uses its own notation, so we will change notation as we switch from approach to approach.

It should be noted that these paradigms are not mutually exclusive: for example, a problem that is modeled using Granger causality can be put into Pearl's SEM if the underlying processes operate in discrete time. Rubin causality can often be phrased as an SEM problem, but in applications this will require more structural assumptions to learn the causal structure. A full exposition of the intersections and non-intersections of these three paradigms is outside the scope of this paper, but we note that these paradigms can often model the same phenomena and shed different insights on the causal behaviors observed.

Rubin causality was first introduced in~\cite{Rubin1974}. In the basic formulation of Rubin causality, we are given some control variable $X$ taking values in $\{0,1\}$. There are also two distinct random variables $Y_0$ and $Y_1$. If $X = 0$, then we observe $Y_0$ and not $Y_1$. If $X = 1$, then we only observe $Y_1$, and not $Y_0$. Another way to write this notationally is that we observe $Y_X$ but do not observe $Y_{1 - X}$, which is often called the counterfactual. The fact that we can only observe one or the other, but not both, is the fundamental misery of causality. 

One of the key results that the Rubin causality paradigm provides is that if $X$ is independent of $Y_0$ and $Y_1$, then randomly assigning $X \in \{0,1\}$ yields a dataset that can provide valid estimates of the counterfactuals; thus, Rubin causality provides the theoretical foundation for randomized control trials. 
This paradigm has also been extended to consider many covariates~\cite{Imbens2015}, handle confounding variables and incorporate instrumental variables~\cite{Imbens2015}, and incorporate some machine learning approaches~\cite{Athey2015}. Sample applications include estimating the causal effect of residential demand response in the Western United States~\cite{Zhou2016} or the causal effects of providing money, healthcare and education to the very poor in Ethiopia, Ghana, Honduras, India, Pakistan, and Peru~\cite{Banerjee1260799}.

Granger causality was first introduced in~\cite{Granger1969}. In this paradigm, we are given data from two stationary random processes $X$ and $Y$, both indexed by time. First, let $U_t$ denote all the information available in the universe at time $t$, and let $(U - X)_t$ denote all the information available at time $t$ except for $X$. Then, let $\sigma^2(Y \vert U)$ denote the error variance of the unbiased, least-squares estimator of $Y_t$ using $U_t$, and similarly let $\sigma^2(Y \vert U - X)$ denote the error variance of the unbiased, least-squares estimator of $Y_t$ using $(U - X)_t$. Then, $X$ Granger-causes (or G-causes, for short) $Y$ if $\sigma^2(Y \vert U) < \sigma^2(Y \vert U - X)$, i.e. the estimator that utilizes $X$ has lower variance on its error than the one that cannot. In other words, $X$ has explanatory power for $Y$.

Granger causality essentially relies on the relationship between causal effects and the arrow of time to distinguish it from general correlations. Although this framework does not address many of the more pernicious philosophical aspects of causality, oftentimes prior knowledge allows us to make the inductive leap from time-lagged correlations to causality. This paradigm is particularly appealing because it is easy to calculate in practice. Sample applications include determining which neuron assemblies Granger-cause other neuron assemblies to fire synapses~\cite{Brovelli2004} or finding that exchange rates Granger-cause stock market prices in Asia~\cite{Granger2000}.

Pearl's SEM approach to causality models the statistical relationship between random elements with a Bayesian network~\cite{Pearl2009}. Bayesian networks are directed acyclic graphs, such that the distribution of a random element at node $i$ only depends on the values taken at the parent nodes. This is meant to model causal relationships between nodes in the graph. Pearl defines the imputation operator as follows: if one imputes at a node $i$, one disconnects $i$ from all its parents and deterministically sets its value to some fixed, predetermined constant. We will be building on this approach in this paper, so we will defer the formal development of Pearl's SEM until Section~\ref{sec:model}. 

At a high level, the imputation operator captures a lot of our intuitions about how the subjunctive conditional should function. When one says \emph{If it had rained today, I would have brought my umbrella}, what does one mean? Intuitively, one often means: `If everything else were the same, only it is the case that it is raining today instead of sunny, these are the actions I would have taken.' One does not mean that the world is structured in a way such that the necessary processes to induce rain today were instead the case. In other words: causal imputation does not travel upstream, e.g. backwards through time. This is captured in Pearl's SEM.

More practically, consider the question: \emph{What are the causal effects of this medication?} If we wish to estimate this, we should `set' medication taken to \texttt{TRUE}, and see the consequences of this imputation. If we do not explicitly `set' this value, then the decision to take medication is a consequence of preceding factors. This makes it difficult to determine if the observed effects are a result of the medication or some other confounding variables\footnote{We note that similar reasoning can be done in the Rubin causality formulation as well.}. Again, this will be more formally discussed in Section~\ref{sec:model}.

Thus, we can think of these paradigms in terms of the central phenomenon it is designed to model. In summary:
\begin{enumerate}
\item Rubin causality is focused on the estimation of the \emph{counterfactual}. 
\item Granger causality is focused on the \emph{explanatory power} one process provides over another process. 
\item Pearl's SEM is focused on the causal effects of the \emph{imputation} operator.
\end{enumerate}

Throughout this paper, we use Pearl's SEM. However, we note again that oftentimes problems framed in the Rubin causality or Granger causality paradigm often can be translated to an equivalent formulation in SEM.

\subsection{Notation}

For any set $A$, we denote the powerset of $A$ as $2^A$, which can also be thought of as the set of functions mapping $A \rightarrow \{0,1\}$. For a collection of sets $\{A_i\}_{i \in I}$, we denote the Cartesian product as $\prod_{i \in I} A_i$.

Also, $I$ will denote the identity matrix, where context will often be sufficient to determine its dimensions.

We let $U[a,b]$ denote the uniform distribution on the interval $[a,b]$ and $N(\mu,\Sigma)$ to denote the multivariate Gaussian distribution with mean $\mu$ and covariance matrix $\Sigma$.

\section{Causal Framework}
\label{sec:model}


In this section, we introduce our framework for modeling causal effects, and then define the problem of optimal causal imputation.

\subsection{Causal structure}

We build on the structural equation modeling framework presented in~\cite{Pearl2009}. First, we will introduce Bayesian networks.

\begin{definition}
A \emph{directed graph} $G = (V,E)$ is a set of nodes $V$ and a set of edges $E \subset V \times V$. Throughout this paper we will assume $V$ is at most countably infinite. 

A \emph{path} from $v_{0} \in V$ to $v_{N} \in V$ is a finite sequence of edges $(v_0, v_1), (v_1, v_2), \dots, (v_{N-1},v_N) \in E$.


We define the \emph{parents} of node $i$ as $\pa(i) = \{ j : (j,i) \in E \}$.

We can iterate this relationship to define the ancestor relationship: let $\pa^n(i) = \{ j : k \in \pa^{n-1}(i), (j,k) \in E \}$, where $\pa^1(i) = \pa(i)$ defined above. Then, the \emph{ancestors} of a node $i$ are given by $\anc(i) = \cup_{n = 1}^\infty \pa^n(i)$.

We say $j$ is a \emph{descendant} of $i$ if $i \in \anc(j)$.

A directed graph is \emph{acyclic} if $i \notin \anc(i)$ for every $i \in V$. We will refer to such graphs as \emph{directed acyclic graphs (DAGs)}.
\end{definition}

\begin{definition}
A \emph{random process} $X$ indexed by a set $V$ is a collection of random elements $(X_i)_{i \in V}$. We will let $\Xc_i$ denote the possible values of $X_i$, and $\Xc = \prod_{i \in V} \Xc_i$.

When there is an associated graph $G = (V,E)$, we will use the notation $\pa(X_i)$ to denote the tuple $(X_j)_{j \in \pa(i)}$.
\end{definition}

\begin{definition}
\label{def:factorization}
A random process $X$ indexed by $V$ is \emph{Markov relative} to a DAG $G = (V,E)$ if its distribution factorizes:
\[
P(X) = \prod_{i \in V} P(X_i \vert \pa(X_i) )
\]

We can also say that $X$ and $G$ are \emph{compatible}, or  $G$ \emph{represents} $X$.
\end{definition}

This formalization will serve as our model for causality. The interpretation is that if there is an edge going from $i$ to $j$, then $X_i$ \emph{causes} $X_j$.

Throughout this paper, we will treat the causal structure $G = (V,E)$ as given. Estimation of this causal structure is a non-trivial task, and an active topic of research. Some approaches to the task of causal inference include: using metrics like directed information to estimate the causal strength between random variables~\cite{Gourieroux1987,Amblard2012}, graphical-model based methods for estimating structure between random variables~\cite{Pearl1998, Lauritzen2001, Li2015, Athey2015a}, and regression based approaches~\cite{Dawid2000,Heckerman2006,Hoyer2009}. 
Again, this list is far from exhaustive as an extensive literature review of this general field is outside the scope of this paper. For a broader overview of various approaches to the problem of causal inference, see~\cite{Pearl1998,Pearl2009}.

Although the estimation of causal structures is never a simple task, the growing field of research promises more and more applications in which accurate estimation of causal structures is feasible.

Previous work has focused on the \emph{estimation} of causal structures. In contrast, our contribution is to consider the problem of \emph{control} of causal structures. In other words, once we are given a causal structure, how can we impute causal effects to drive the overall system into a desirable state?

For example, once we can estimate the causal effects of issuing rebates for energy-efficiency appliances, how do we best distribute these rebates to induce more energy-efficient consumption patterns? To the best of our knowledge, this is the first paper to consider the problem of when and where to impute on a causal structure.

There is an equivalent formulation of the condition in Definition~\ref{def:factorization} which utilizes \emph{disintegration} results in probability theory. This is referred to as the structural equation modeling framework in~\cite{Pearl2009}.

\begin{proposition} \cite{Kallenberg2002, Pearl2009}
A random process $X$ indexed by $V$ is Markov relative to $G = (V,E)$ if and only if there exists a collection of functions $(f_i)_{i \in V}$ and independent random elements $(\xi_i)_{i \in V}$ such that:
\begin{equation}
\label{eq:sem}
X_i = f_i(\pa(X_i),\xi_i)
\end{equation}

Furthermore, if $\Xc_i$ are Borel spaces\footnote{A measurable space $S$ is Borel if there exists a measurable function $S \rightarrow [0,1]$ with a measurable inverse.}, then $\xi_i$ can be taken to be $U[0,1]$.
\end{proposition}

We note that Borel spaces are a very general category of measurable spaces: they include Polish spaces equipped with the Borel $\sigma$-algebra\footnote{A topological space $T$ is Polish if it is separable and completely metrizable. The Borel $\sigma$-algebra of a topological space is the smallest $\sigma$-algebra containing all the open sets.}. This includes finite sets, $\R$, $\R^n$, $L^p(\R^n)$, the set of $p$-integrable functions defined on $\R^n$. Additionally, the space of probability distributions on any Borel space is also a Borel space.

\begin{assumption}
Throughout the rest of this paper, we will always use $X$ to denote a random process indexed by $V$ that is Markov relative to a DAG $G = (V,E)$, where $X_i$ takes values in $\Xc_i$. Similarly, $f_i$ shall denote the functions as specified in Equation~\ref{eq:sem}, and similarly $\xi_i$.
\end{assumption}

\subsection{Causal imputation}

In this section, we will formally define the \emph{causal imputation} operation. Intuitively, imputation of $X$ produces a new random process $Y$.  This random process $Y$ is equal to $X$ prior to the causal imputation, is forced to some value at the node of imputation, and experiences causal effects after the node of imputation. This is formally defined below.

\begin{definition} \cite{Pearl2009}
A random process $Y$ indexed by $V$ is the \emph{imputation} of $X$ at $i \in V$ to a constant $x_i \in \Xc_i$ if:
\begin{itemize}
\item $Y_i = x_i$.
\item For any $j$ that is not a descendant of $i$, $Y_j = X_j$.
\item For any $j$ that is a descendant of $i$, $Y_j = f_j(\pa(Y_j),\xi_j)$.
\end{itemize}

If this is the case, we will write $Y = \impute(X;i,x_i)$.
\end{definition}

The imputation operator produces a copy of the original process that is exactly equal at all nodes that do not causally depend on the node of imputation $X_i$. At the point of imputation, the node is disconnected from its parents and forced a constant value $x_i$. The nodes $X_j$ that causally depend on $X_i$ are replaced with new values that depend on the causal effects of $X_i$, keeping the innovation terms $\xi$ constant throughout.

Referring back to the discussions in Section~\ref{sec:back}, this can be thought of as manually setting the value of $X_i$ to $x_i$. This should be something that is done exogenously, as a control variable, rather than as a consequence of endogenous factors: this is why $Y_i$ is disconnected from $\pa(Y_i)$.

From this definition, it immediately follows that the imputation operator commutes.

\begin{proposition}
\label{prop:commute}
Let $i, j \in V$ such that $i \neq j$ and $x_i \in \Xc_i$ and $x_j \in \Xc_j$. Then $\impute( \impute(X;i,x_i); j,x_j ) = \impute( \impute(X;j,x_j) ; i,x_i)$ almost surely.
\end{proposition}

This allows us to define imputation on any set of nodes, rather than just at a single node.

\begin{definition}
For any $I \subset V$ and $x_I \in \prod_{i \in I} \Xc_i$, we define the imputation $Y = \impute(X;I,x_I)$ as the sequential application of element-wise $\impute$ operations. This is almost surely unique by Proposition~\ref{prop:commute}.
\end{definition}

\subsection{Optimal causal imputation}

In the previous section, we defined the causal imputation operator. We can think of our system designer as having the capacity of issuing control commands that have causal effects on the system downstream. When we can define the cost of imputation as well as a control objective, we can formulate the optimal causal imputation problem.

We suppose we are given a collection of functions $(c_I)_{I \subset V}$ where each $c_I : \prod_{i \in I} \Xc_i \rightarrow \R$. These functions can be interpreted as the cost of imputation at a set of nodes $I \subset V$. Drawing on our running example,  $c$ represents the cost of issuing rebates for eco-friendly refrigerators at a set of households.

Furthermore, we suppose we are given an operational objective in the form of a cost function $g : \Xc \rightarrow \R$. For example, $g$ can be a penalty on energy-wasting consumption patterns.

\begin{definition}
\label{def:optimal_causal_imputation}

The problem of \emph{optimal causal imputation} is given by:
\begin{eqnarray}
\min_{I \subset V} \min_{x_I \in \prod_{i \in I} \Xc_i}  & c_I(x_I) + \E_Y[g(Y)] \\
\mathrm{subject~to~} & Y = \impute(X;I,x_I)
\end{eqnarray}
\end{definition}

\section{Applications}
\label{sec:app}


In Section~\ref{sec:model}, we defined the optimal causal imputation problem in its full generality. In this section, we shall provide methods to solve the optimal causal imputation problem in special cases. In particular, we consider two contexts: 1) situations where imputation is only allowed to a single value, 2) situations where the dynamics are linear-Gaussian. In both instances, we shall assume $\Xc_i = \R^{n_i}$ for some $n_i$.

\subsection{Single-value case}

In many applications where we can causally impute values, we can only impute to one particular value. For example, when issuing incentives, we may be able to only offer one form of rebate to consumers. Motivated by this context, we consider situations where the optimal causal imputation problem can be reduced to one of submodular optimization.

\begin{assumption}
In this section, we assume $V$ is a finite set and that for each $I \subset V$, there exists an $x_I$ such that $c_I(x_I) < \infty$ and $c_I(x_I') = \infty$ for any $x_I' \neq x_I$. We shall refer to this as the single-value case.

In the single-value case, we use the shorthand $F(I) = c(I) + \E[g(\impute(X;I))]$, where we drop dependencies on $x$ as it can only take a single value.
\end{assumption}

\subsubsection{Submodular minimization}

\begin{definition}
The set mapping $F : 2^V \rightarrow \R$ is \emph{submodular} if for any $I_1 \subset I_2 \subset V$ and $i \in G \setminus I_2$, we have:
\begin{equation}
\label{eq:submodular}
F(I_1 \cup \{i\}) - F(I_1) \geq F(I_2 \cup \{i\}) - F(I_2)
\end{equation}
\end{definition}
Intuitively, this definition is motivated by economies of scale. We often expect economies of scale from these imputations, e.g. the per-customer cost of a rebate is non-increasing as the number of customers increases, due to bulk-purchase discounts. In our running example, the additional cost of issuing a rebate to customer $i$ is higher when you have issued few rebates than when you have issued a lot of rebates. (When $I_1 \subset I_2$, then $I_2$ corresponds to the situation where you have issued more rebates than $I_1$.)

From a combinatorial optimization perspective, submodularity is a very well-behaved property that makes optimization, or approximate optimization, very tractable. We shall quickly outline the details now, but we refer the interested reader to~\cite{Schrijver2003} for more details.

First, note that there is a very direct correspondence between a subset $I \subset V$ and a tuple in $\{0,1\}^V$. For example, if $V = \{0,1,2\}$, then $(0,1,1)$ corresponds to the subset $\{1,2\}$. Thus, we can think of $F : \{0,1\}^V \rightarrow \R$. Now, we define the Lov\'asz extension~\cite{Lovasz1983}.

\begin{definition}
Let $\lambda \sim U[0,1]$. Then, for any set mapping $F : \{0,1\}^V \rightarrow \R$, we define the \emph{Lov\'asz extension} $f : [0,1]^V \rightarrow \R$ as:
\[
f(z) = \E_\lambda[ F( \{ i : z_i > \lambda \} ) ]
\]
For the rest of this section, an unindexed $f$ will denote the Lov\'asz extension of $F$.
\end{definition}

We note two nice properties of the Lov\'asz extension immediately.

\begin{proposition} \cite{Lovasz1983}
For any $z \in \{0,1\}^V$, we have $f(z) = F(z)$. 
\end{proposition}

\begin{proposition} \cite{Lovasz1983}
$F$ is submodular if and only if $f$ is convex.
\end{proposition}

Note that the optimal causal imputation problem can be written as:
\[
\min_{z \in \{0,1\}^V} F(z)
\]
The Lov\'asz extension provides us with an easy solution to the problem.

\begin{proposition} \cite{Lovasz1983}
If $F$ is submodular, then the following is a convex optimization program.
\begin{equation}
\label{eq:lovasz_relax}
\min_{z \in [0,1]^V} f(z)
\end{equation}
Furthermore, there exist minimizers of~\eqref{eq:lovasz_relax} in $\{0,1\}^V$.
\end{proposition}

In other words, the combinatorial optimization problem can be solved tractably with convex optimization if $F$ is submodular. Thus, we are motivated in searching for conditions under which $F(I) = c(I) + \E[g(\impute(X;I))]$ is submodular. We provide a common sufficient condition for submodularity of $F$ in the following theorem:
\begin{theorem}
\label{theorem:submodular}
If:
\begin{itemize}
\item $g(Y) = \|Y_i - \E Y_i \|_2^2$ for some $i \in V$.
\item There exists functions $f_j^\xi$ such that, if $X_j$ has no parents, $X_j = f_j^\xi(\xi_j)$ and otherwise $X_j = \pa(X_j) + f_j^\xi(\xi_j)$.
\item For each $j \in \anc(i)$, there exists one unique path from $j$ to $i$.
\item $c(I)$ is submodular.
\end{itemize}
Then $F(I) = c(I) + \E[g(\impute(X;I))]$ is submodular.
\end{theorem}

Note here that we treat $\pa(X_i)$ as a vector in $\R^{n_i}$, where $n_i$ is the appropriate dimension. These assumptions encompass many graphical models where a node's parents set a location parameter, and the control objective is the second moment of some feature.

\begin{proof}
Note that the desired result will follow if we show that the set mapping $G : I \mapsto \E[g(\impute(X;I))]$ is submodular, since the sum of submodular functions is submodular. Throughout this proof, we use $i$ to refer to the index $i$ pulled out by the function $g$.

We can see that $G(\emptyset) = \E[g(X)]$. By the independence of the $(\xi_i)_{i \in V}$ and the form of the $(X_i)_{i \in V}$ relationships, we can write this as $\E[g(X)] = \sum_{j \in \anc(i)} \| f_j^{\xi}(\xi_j) - \E f_j^{\xi}(\xi_j) \|_2^2$. (Note that the unique path assumption ensures that each variance is only counted once in this sum.)

More generally, we can write an expression for $G(I)$. Note that if we impute at a node $j$, all the uncertainty due to node $j$, and the ancestors of $j$, is zeroed out. Thus, we can write $G(I) = \E[g(X)] - \sum_{j \in (I \cup \anc(I))} \| f_j^{\xi}(\xi_j) - \E f_j^{\xi}(\xi_j) \|_2^2$, where we define $\anc(I) = \cup_{j \in I} \anc(j)$.

Now, we can verify the submodularity condition on $G$. Pick $I_1 \subset I_2$ and $i' \in V \setminus I_2$. Then:
\[
G(I_1 \cup \{i'\}) - G(I_1) =
\]
\[
\sum_{j \in (I_1 \cup \anc(I_1))}  \| f_j^{\xi}(\xi_j) - \E f_j^{\xi}(\xi_j) \|_2^2 - 
\]
\[
\sum_{j \in (I_1 \cup \{i'\} \cup \anc(I_1 \cup\{i'\}))}  \| f_j^{\xi}(\xi_j) - \E f_j^{\xi}(\xi_j) \|_2^2 =
\]
\[
- \sum_{j \in \{i'\} \cup (\anc(i') \setminus \anc(I_1))}  \| f_j^{\xi}(\xi_j) - \E f_j^{\xi}(\xi_j) \|_2^2
\]
In words, the change in $G$ due to adding $i'$ to $I_1$ is the variances due to the terms related to $i'$ and the ancestors of $i'$ that have not already been zeroed out due to imputation, i.e. the ancestors of $i'$ that are not already ancestors of $I_1$. A similar derivation can be done for $I_2$.

Thus, we can verify that $G(I_1 \cup \{i'\}) - G(I_1) \geq G(I_2 \cup \{i'\}) - G(I_2)$ by noting that $\anc(i') \setminus \anc(I_2) \subset \anc(i') \setminus \anc(I_1)$, so the right-hand side of the inequality adds more negative terms. This concludes our proof.
\end{proof}

\subsubsection{Submodular maximization}

Alternatively, suppose we are attempting to maximize a submodular function subject to a constraint, i.e. $F(I) = c(I) + \E[g(\impute(X;I))]$ subject to a constraint that $I \in S \subset 2^V$ and our objective is to solve $\max_{I \in S} F(I)$.\footnote{Strictly speaking, to remain consistent with the problem in Section~\ref{sec:model}, we should be solving $\min_{I \in S} - F(I)$, but we express it as a maximization for clarity of presentation.}

First, consider the greedy method for submodular maximization. This is presented as Algorithm~\ref{alg:sub_max}. At each iteration, it simply adds an element to $I$ which maximizes $F(I \cup \{i\})$, if one exists. If one does not exist, it terminates and returns $I$. Under certain structural conditions, this algorithm yields approximate optimizers.

\begin{algorithm}[t]
\begin{algorithmic}
\STATE $I \gets \emptyset$
\WHILE{ $\max_{i : I \cup \{i\} \in S} F(I \cup \{i\}) - F(I) \geq 0$ }
\STATE Pick $i^* \in \arg\max_{i : I \cup \{i\} \in S} F(I \cup \{i\})$
\STATE $I \gets I \cup \{i^*\}$
\ENDWHILE
\RETURN $I$
\end{algorithmic}
\caption{The greedy approach for combinatorial maximization.}
\label{alg:sub_max}
\end{algorithm}

\begin{definition}
A set mapping $F : 2^V \rightarrow \R$ is nondecreasing if $F(S) \leq F(T)$ whenever $S \subset T$.
\end{definition}

The monotonicity condition effectively prevents the algorithm from straying too far from the optimum when taking the greedy approach, as shown in~\cite{Nemhauser1978}. Note that if $F$ is non-decreasing, then the condition $\max_{i : I \cup \{i\} \in S} F(I \cup \{i\}) - F(I) \geq 0$ is equivalent to the existence of $i \in V$ such that $I \cup \{i\} \in S$.

\begin{proposition} \cite{Nemhauser1978}
If $F$ is nondecreasing and submodular, then the greedy method presented in Algorithm~\ref{alg:sub_max} will return $I^* \in S$ such that $F(I^*) \geq \left( \frac{e-1}{e} \right) \max_{I \in S} F(I)$.
\end{proposition}

We now present a quick corollary of Theorem~\ref{theorem:submodular}, which provides conditions under which we can leverage the existing results for maximization of nondecreasing submodular functions.

\begin{corollary}
\label{cor:sup}
If:
\begin{itemize}
\item $g'(Y) = - \|Y_i - \E Y_i \|_2^2$ for some $i \in V$.
\item There exists functions $f_j^\xi$ such that, if $X_j$ has no parents, $X_j = f_j^\xi(\xi_j)$ and otherwise $X_j = \pa(X_j) + f_j^\xi(\xi_j)$.
\item For each $j \in \anc(i)$, there exists one unique path from $j$ to $i$.
\item $c(I)$ is nondecreasing and submodular.
\end{itemize}
Then $F(I) = c(I) + \E[g'(\impute(X;I))]$ is nondecreasing and submodular.
\end{corollary}
\begin{proof}
This follows from Theorem~\ref{theorem:submodular} if we can show that $G': I \mapsto \E[g(\impute(X;I))]$ is nondecreasing. Let $Y = \impute(X;I)$, and note that adding elements to $I$ can only decrease the variance of $Y_i$. This can be formalized by noting, similar to the arguments in the proof of Theorem~\ref{theorem:submodular}, $G'(I) = \E[g'(X)] +\sum_{j \in (I \cup \anc(I))} \| f_j^{\xi}(\xi_j) - \E f_j^{\xi}(\xi_j) \|_2^2$. Thus, $G'$, the additive inverse of the variance of $Y_i$, is nondecreasing.
\end{proof}

Note the minus sign in $g'$ in Corollary~\ref{cor:sup}: in most instances where you are maximizing a submodular cost, you would still wish to reduce uncertainty, i.e. have a lower variance.

\subsection{Linear-Gaussian case}

In this section, we consider causal imputation on a discrete-time linear dynamical system with Gaussian noise. That is, we analyze the special case of a random process with the form:
\[
X_{t+1}=AX_t+\epsilon_t
\]
Where $X_t \in \R^n$, $\epsilon_t \sim N(0,\sigma^2I)$ independently for $t=0,...,T$, and $A \in \R^{n \times n}$ is a matrix representing the dependencies.

This process can be represented as a causal graph in the form of a trellis, where the random variables are all Gaussian. More specifically, each node has its expected value equal to a linear combination of their parents, as described by a matrix $A$, and additive noise of the distribution $N(0,\sigma^2)$.

To analyze our optimal casual imputation problem, we first redefine the indices for this problem. Since our causal graph represents a process over time, we index into the process by state $k$, for $k=1,...n$ as well as a time $t$ for $t=0,...,T$. Thus $X_{kt}$ indicates the value of state $k$ at time $t$, and our graph has vertices $V = \{1,\dots,n\} \times \{0,\dots,T\}$. As before, $X_t$ represents the value of the vector of all the states of $X$ at time $t$, and we can think of $X$ as a vector in $\R^{nT}$. We assume that the cost of imputation $c_I(x_I)$ has the following form for some parameters $\delta_i, q_i \geq 0$:
\[c_I(x_I)=\sum_{i \in I} \delta_i + q_ix_i^2\]

Further, we look at the case where the system cost of interest is minimizing the expected distance of the the random process from some target trajectory $\bar{y}$. Thus $g(Y)=\|Y-\bar{y}\|^2_2$.

Our optimal causal imputation problem in this case is thus:
\begin{eqnarray}
\begin{split}
\min_{S \subset V} \min_{x_S \in \R^S}  & \sum_{i \in S} \left( \delta_i + q_i x_i^2 \right) + \E\left[ \| Y - \bar{y}\|_2^2 \right] \\
\mathrm{subject~to~} & Y = \impute(X;S,x_S)
\end{split}
\label{eqn:opt_gauss}
\end{eqnarray}
The summation term can be thought of as a cost of issuing control commands and the expectation term can be thought of as a trajectory tracking objective.

Given our structure on the random process, we can rewrite this optimization problem more concretely.

We first define $Q \in \mathbb{R}^{nT \times nT}$ to be diagonal matrix with the $q_i$'s on the diagonal. We define $\delta \in \mathbb{R}^{nT}$ to be the vector of $\delta_i$'s. Further, let $\bold{1}_{nT}$ denote the column vector of all ones in $\R^{nT}$.
Lastly, we define $\diag(S)$ to be the square matrix with the elements of $S$ on the diagonal, and zeros everywhere else.

The optimization in~\eqref{eqn:opt_gauss} now becomes:
\begin{eqnarray*}
\min_{\substack{S \in \{0,1\}^{nT} \\ \bar{x} \in \mathbb{R}^{nT}}} &
\bar{x}^\T(Q+D)\bar{x} +\sigma^2\trace(DI_S)+\delta^\T S-2\bar{y}^\T\bar{x} \\
\mathrm{subject~to~} &
(S_i - 1) \bar{x}_i=0 ~ \text{ for all } i=1, \dots ,nT\\
&
P = (1 - \tilde{A})^{-1} \\
&
D = P^\T P \\
& 
I_S = I-\diag(S)
\end{eqnarray*}
\[
\tilde{A}= I_S\begin{bmatrix}
    0       & 0 & 0 & \dots & 0 & 0 \\
    A       & 0 & 0 & \dots & 0 & 0 \\
    0       & A & 0 & \dots & 0 & 0 \\
   0 &0 &A &\ddots &0 & 0\\
   \vdots &\vdots &\vdots &\ddots &\vdots& \vdots\\
    0       & 0 & 0 & \dots & A & 0
\end{bmatrix}
\]
We note that for any matrix $A$ and any $S$, the matrix $I-\tilde{A}$, with $\tilde{A}$ as defined above, is invertible, so $P$ will always be well-defined. 

Additionally, for a fixed $S$, the optimization across $\bar{x}$ is easy to solve. That is, $D$ is entirely determined by $S$. If we let $(Q+D)_S$ denote the submatrix of $(Q+D)$ indexed by the non-zero elements of $S$, and similarly $\bar{x}_S$ and $\bar{y}_S$, then the optimizer is given by $\bar{x}_S^* = (Q+D)_S^{-1} \bar{y}_S$, with the other entries of $\bar{x}^*$ equal to $0$.

Thus, we can easily calculate a set mapping $F(S)$ such that optimal causal imputation in the linear-Gaussian case is simply $\min_{S \subset V} F(S)$. We can solve this when $nT$ is relatively small, and are currently investigating properties of $F(S)$ which would allow us to apply combinatorial optimization techniques~\cite{Schrijver2003}.

\section{Conclusion and Future Work}
\label{sec:conclusions}


The previous literature on mathematical formulations of causality has been focused on the \emph{estimation} of causal structures. In this paper, we presented the problem of \emph{control} of causal structures. 
We formally defined the problem of optimal causal imputation, and formulate solutions for it in two cases: where imputation is allowed to only a single value, and the case where the dynamics are linear and the noise is Gaussian. 

In future work, we hope to apply this framework to real situations which allow both the estimation of causal structures, as well as verification of the consequences and costs of imputation. Additionally, we hope to generalize our results to consider dynamical systems whose behavior are influenced by different features. For example, we can consider the dynamics of the power grid, but also account for frequently used machine learning features as well, such as the zip code of different energy consumers and the age of deployed assets. 

We believe that considering the control aspects of causality is increasingly more relevant. In many smart infrastructure applications, we no longer have control commands that directly affect the dynamics, but rather our control actions act more like causal imputations.  
The optimal causal imputation framework is a promising direction to model these interactions between machine learning and control, and provides a model for closing the loop on analytics in cyber-physical systems.

\bibliographystyle{IEEEtran}
\bibliography{DONG_ROY_refs}

\end{document}